\documentclass{llncs}

\usepackage{amsmath,amssymb,latexsym}
\usepackage{booktabs}
\usepackage{xspace}
\usepackage{centernot,cancel}
\usepackage{paralist}
\usepackage{url}
\usepackage[usenames]{color}
\usepackage{caption}
\usepackage{subcaption}
\usepackage[caption=false]{subfig}
\usepackage{tikz}
\usetikzlibrary{positioning,arrows,automata,shapes,shadows,fit}

\newcommand{\set}[1]{\left\{#1\right\}}
\newcommand{\card}[1]{\left|#1\right|}

\newcommand{\uset}[1]{\ensuremath{\mathord{\uparrow}#1}}

\newcommand{\dset}[1]{\ensuremath{\mathord{\downarrow}#1}}
\newcommand{\defeq}{\mathrel{\stackrel{\rm def}{=}}}

\newcommand{\etal}{\emph {et al.}\xspace}

\newcommand{\nlessdot}{\mathrel{\centernot\lessdot}}

\newcommand{\keysp}{\mathcal{K}}

\newcommand{\eofg}{E_0^*}
\newcommand{\eofh}{E_0}

\newcommand{\vin}[1][x]{#1_{\rm in}}
\newcommand{\vout}[1][x]{#1_{\rm out}}

\newcommand{\treepi}{\widetilde{\Pi}}

\newcommand{\setup}{\mathsf{SetUp}}
\newcommand{\derive}{\mathsf{Derive}}

\begin{document}
 
\title{Optimal Constructions for\\ Chain-based Cryptographic Enforcement\\ of Information Flow Policies}
\author{Jason Crampton \and Naomi Farley \and Gregory Gutin \and \mbox{Mark Jones}}
\institute{Royal Holloway, University of London}
%   \author{}\institute{}

\maketitle

\begin{abstract}
The simple security property in an information flow policy can be enforced by encrypting data objects and distributing an appropriate secret to each user.
A user derives a suitable decryption key from the secret and publicly available information.
A chain-based enforcement scheme provides an alternative method of cryptographic enforcement that does not require any public information, the trade-off being that a user may require more than one secret.
For a given information flow policy, there will be many different possible chain-based enforcement schemes.
In this paper, we provide a polynomial-time algorithm for selecting a chain-based scheme which uses the minimum possible number of secrets.
We also compute the number of secrets that will be required and establish an upper bound on the number of secrets required by any user.
\end{abstract}

\section{Introduction}\label{sec:intro}

Access control is a fundamental security service in modern computing systems and seeks to restrict the interactions between users of the system and the resources provided by the system.
Generally speaking, access control is policy-based, in the sense that a policy is defined by the resource owner(s) specifying those interactions that are authorized.
An attempt by a user to interact with a protected resource, typically called an \emph{access request}, is evaluated by a trusted software component, the \emph{policy decision point} (or \emph{authorization decision function}), to determine whether the request should be permitted (if authorized) or denied (otherwise).
The use of a policy decision point is entirely appropriate when we can assume the policy will be enforced by the same organization that defined it.
However, use of third-party storage, privacy policies controlling access to personal data, and digital rights management all give rise to scenarios where this assumption does not hold.

An alternative approach to policy enforcement, and one that has attracted considerable interest in recent years, is to encrypt the protected object and enable authorized users to derive decryption keys.
%issue appropriate decryption keys to authorized users.
This approach is particularly suitable for data that changes infrequently, for read-only policies, and for policies that can be represented in terms of user attributes.
Research into cryptographic access control began with the seminal work of Akl and Taylor~\cite{AkTa83} on the enforcement of information flow policies, and has seen a resurgence of interest in recent years.

Generally, it is undesirable to provide a user with all the keys she requires to decrypt protected objects.
Instead, a user is given a small number of secrets from which she is able to derive all keys required.
Thus a cryptographic enforcement scheme may be characterized by%
\begin{inparaenum}[(i)]
 \item the number of secrets each user has to store,
 \item the total number of secrets,
 \item the amount of auxiliary (public) information required for key derivation, and
 \item the amount of time required for key derivation.
\end{inparaenum}

Many schemes in the literature provide each user with a single secret~\cite{atal:dyna09,cram:key06}, the trade-off being that the amount of public information and derivation time may be substantial.
In contrast, chain-based schemes require no public information but each user may require more than one secret~\cite{CrDaMa10,frei:prov11,frei:simp13}.
In addition, chain-based schemes can achieve very strong security properties~\cite{frei:simp13}.
There are many different ways to instantiate a chain-based scheme for a given policy, each instantiation being defined by a chain partition of the partially ordered set that defines the policy.

However, existing work on chain-based CESs assumes the existence of a chain partition and simply generates the required secrets and keys for this partition~\cite{CrDaMa10,frei:prov11,frei:simp13}.
This approach ignores the fact that there will be (exponentially) many choices of chain partition.
Thus, it is important, if we are to make best use of chain-based CESs, that we know which chain partition to use for a given information flow policy.
It is this issue that we address in this paper.

\paragraph{Contributions.}
Our first contribution (Theorem~\ref{thm:keys-equal-sum-of-omega}) is to show how $\widehat{K}(\Pi)$, the (total) number of secrets for a chain partition $\Pi$, is related to the set of edges in the representation of $\Pi$ as an acyclic directed graph.
We then prove that $\widehat{K}(\Pi)$ is determined by the end-points of the chains in $\Pi$ (Lemma~\ref{lem:number-of-keys-from-bottom-elements}).
This, in turn, allows us to prove there exists a chain partition that simultaneously minimizes the number of secrets required and the number of chains in the partition (Theorem~\ref{thm:chain-partition-only-requires-w-chains}).
The last result is somewhat unexpected, as it is not usually possible to simultaneously minimize two different parameters.
The result is also of practical importance, since the number of chains in $\Pi$ provides a tight upper bound on the number of secrets required by any one user.
Our main contribution (Theorem~\ref{thm:main-theorem} and Section~\ref{sec:chain-partition-requiring-widehat-k-keys}) is to develop a polynomial-time algorithm that enables us to find a chain partition $\Pi$ such that $\widehat{K}(\Pi)$ and the number of chains is minimized (with respect to all chain partitions).
Our algorithm is based on finding an optimal feasible flow in a network and makes use of the characterization of the number of secrets in terms of the set of edges (established in Theorem~\ref{thm:keys-equal-sum-of-omega}) to define the capacities of the edges in the network.
We thereby provide rigorous foundations for the development of efficient chain-based enforcement schemes.
% 
% However, it is not currently known what the best instantiation is, nor how that instantiation can be computed.
% It is these problems that we address in this paper.
% The main contribution of our work is to provide a polynomial-time algorithm generating an instantiation that ensures the total number of keys ($\widehat{K}$) is minimized and the maximum number of keys required by any one user ($k_{\max}$) is no greater than the width of the partially ordered set.
% Somewhat surprisingly, $\widehat{K}$ and $k_{\max}$ can be minimized at the same time.
% In addition, we provide an elementary method for computing $\widehat{K}$ and $k_{\max}$.
% (This second contribution may provide an efficient method of computing chain partitions in practice.)
% We thereby provide rigorous foundations for the development of efficient chain-based enforcement schemes.

\paragraph{Paper structure.}
In the next section, we provide the relevant background on cryptographic enforcement schemes, formally define the problem, and state Theorem~\ref{thm:main-theorem}.
In Sec.~\ref{sec:computing-k-max-and-widehat-k}, we state and prove Theorems~\ref{thm:keys-equal-sum-of-omega} and~\ref{thm:chain-partition-only-requires-w-chains} and Lemma~\ref{lem:number-of-keys-from-bottom-elements}. %, which connect the structure of the information flow policy to the number of keys required in a chain-based enforcement scheme.
% In particular, we show that there exists a chain-based scheme that simultaneously uses a minimum number of keys in total and a minimum number of chains, thereby imposing an upper bound on the number of secrets required by any one user.
% The results in Sec.~\ref{sec:computing-k-max-and-widehat-k}, however, are not constructive.
In Sec.~\ref{sec:chain-partition-requiring-widehat-k-keys}, we develop an efficient algorithm to derive the best chain partition and prove Theorem~\ref{thm:main-theorem}.
We conclude the paper with a summary of our contributions and some ideas for future work.

\section{Background and Problem Statement}\label{sec:background}

A \emph{partially ordered set} (or \emph{poset}) is a pair $(X,\leqslant)$, where $\leqslant$ is a reflexive, anti-symmetric, transitive binary relation on $X$.
We may write $x \geqslant y$ whenever $y \leqslant x$, and $y < x$ whenever $y \leqslant x$ and $y \ne x$.
Given a poset $(X,\leqslant)$, it is convenient to introduce the following notation.
\[
 \dset{x} \defeq \set{y \in X : y \leqslant x} \qquad\text{and}\qquad
 \uset{x} \defeq \set{y \in X : y \geqslant x}
\]
We will also make use of the following terminology and notation.
\begin{compactitem}
  \item We say $x$ \emph{covers} $y$, denoted $y \lessdot x$, if $y < x$ and there does not exist $z \in X$ such that $y < z < x$.
        We say $y$ is a \emph{child} of $x$ if $y \lessdot x$ (and $x$ is a \emph{parent} of $y$).
  \item The \emph{Hasse diagram} of a poset is the directed acyclic graph $H = (X,\eofh)$, where $xy \in E_0$ if and only if $y \lessdot x$.
  \item $X$ is a \emph{tree} if no element of $X$ has more than one parent and $X$ has a unique maximum element.
  \item $Y \subseteq X$ is a \emph{chain} (or \emph{total order}) if for $x,y \in Y$, $x < y$ or $x = y$ or $y < x$.
	$\set{C_1,\dots,C_\ell}$ is a \emph{chain partition} (of $(X,\leqslant)$) if $C_i \subseteq X$ is a chain, $C_i \cap C_j = \emptyset$ if $i \ne j$, and $C_1 \cup \dots \cup C_\ell = X$.
  \item $Y \subseteq X$ is an \emph{antichain} if for $x,y \in Y$, $x \leqslant y$ if and only if $x = y$.
        (In other words, for $x \ne y$ in an antichain, $x \not\leqslant y$ and $y \not\leqslant x$.)
	The \emph{width} of a poset is the cardinality of an antichain of maximum size. 
\end{compactitem}

An illustrative Hasse diagram is shown in Fig.~\ref{fig:hasse-diagram}.
In the poset depicted, $\set{a,d,f}$ is a chain, for example, and $\set{d,e}$ is an antichain of maximum size.
Thus the width of this poset is $2$ and one chain partition of cardinality $2$ is $\set{\set{a,c,e,g,h},\set{b,d,f}}$.

\begin{figure}[!t]\centering
    \begin{tikzpicture}[v/.style={circle,draw,fill=white,inner sep=0pt,minimum width=4pt},>=stealth,<-,scale=.9, transform shape]
      \node[v,label=left:$a$] (a) {};
      \node[v,above left=of a,label=left:$b$] (b) {};
      \node[v,above right=of a,label=left:$c$] (c) {};
      \node[v,above right=of b,label=left:$d$] (d) {};
      \node[v,above right=of c,label=left:$e$] (e) {};
      \node[v,above left=of d,label=left:$f$] (f) {};
      \node[v,above right=of d,label=left:$g$] (g) {};
      \node[v,above left=of g,label=left:$h$] (h) {};
      \draw (a) -- (b);
      \draw (a) -- (c);
      \draw (b) -- (d);
      \draw (c) -- (d);
      \draw (c) -- (e);
      \draw (d) -- (f);
      \draw (d) -- (g);
      \draw (e) -- (g);
      \draw (f) -- (h);
      \draw (g) -- (h);
    \end{tikzpicture} 
   \caption{The Hasse diagram of a simple poset}\label{fig:hasse-diagram}
\end{figure}
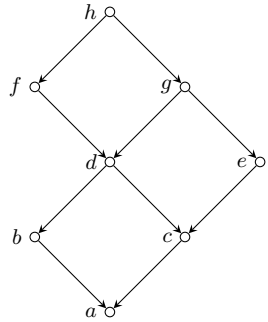

\begin{definition}
An \emph{information flow policy} is a tuple \mbox{$(X,\leqslant,U,O,\lambda)$}, where:
  \begin{compactitem}
    \item $(X,\leqslant)$ is a (finite) partially ordered set of \emph{security labels};
    \item $U$ is a set of \emph{users} and $O$ is a set of \emph{objects};
    \item $\lambda: U \cup O \rightarrow X$ is a \emph{security function} that associates users and objects with security labels.
  \end{compactitem}
The \emph{simple security property} requires that user $u \in U$ can read an object $o \in O$ if and only if $\lambda(u) \geqslant \lambda(o)$.
\end{definition}

We may define an equivalence relation $\sim$ on $U$, where $u \sim v$ if and only if $\lambda(u) = \lambda(v)$.
We write $U_x$ to denote $\set{u \in U : \lambda(u) = x}$; $U$ is partitioned into the set of equivalence classes $\set{U_x : x \in X}$.
Similarly, $O_x \subseteq O$ is the set of objects having security label $x \in X$.
Thus, the simple security property guarantees that any $o \in O_x$ can be read by a user $u \in U_y$ for any $y \geqslant x$.
Conversely, $u \in U_y$ can read $o \in O_x$ for any $x \leqslant y$.
Henceforth, we will represent an information flow policy $(X,\leqslant,U,O,\lambda)$ as a pair $(X,\leqslant)$ with the tacit understanding that $U$, $O$ and $\lambda$ are given.
% We assume that the policy will be implemented by encrypting objects and distributing keys to users, enabling them to decrypt the objects to which they should have access.%

\subsection{Cryptographic Enforcement of Information Flow Policies}

One way of enforcing the simple security property (for policy $(X,\leqslant)$) is to encrypt $o \in O_y$ with a (symmetric) key $k(y)$ and provide all users in $U_x$, where $x \geqslant y$ with the key $k(y)$.
An alternative is to provide a user $u$ in $U_x$ with a smaller number of keys (typically a single key for label $x$) and enable $u$ to derive keys for all $y$ such that $y < x$.
However, this introduces the possibility that users may be able to collude and use their keys to derive a key that no single user could derive.

More formally, there exists the notion of a \emph{cryptographic enforcement scheme} (CES), defined by the $\setup$ and $\derive$ algorithms, $\setup$ being used to generate secrets and keys and the data used to derive secrets and keys, and $\derive$ being used to compute secrets and keys.
Let $\keysp$ denote an arbitrary key space (typically $\keysp=\set{0,1}^l$ for some $l\in\mathbb{N}$). 
Then $\setup$ and $\derive$ have the following characteristics.
\begin{itemize}
 \item $\setup$ takes as input a security parameter~$\rho$ and information flow policy $(X,\leqslant)$.
       
       It outputs, for each element $x \in X$, a pair $(\sigma(x), \kappa(x))$: the \emph{secret} $\sigma(x)$ is given to all users in $U_x$; $\sigma(x)$ is used to derive secrets and/or keys for labels $y \leqslant x$; and the \emph{key} $\kappa(x) \in \keysp$ is used to encrypt data objects in $O_x$.
       
       The $\setup$ algorithm also outputs a set of public information $\sf Pub$, which is used for the derivation of secrets and keys.%
 \item $\derive$ takes as input $(X,\leqslant)$, $\sf Pub$, start and end points $x,y \in X$ and $\sigma(x)$. % changed by bertram
       
       It outputs $\kappa(y)\in\keysp$ if and only if $y \leqslant x$.
       (In particular, $\kappa(x)$ can be derived from $\sigma(x)$.)
\end{itemize}

The requirement that $\derive$ outputs $\kappa(y)$ (given $\sigma(x)$) if $y \leqslant x$ is a correctness criterion, which ensures an authorized user can derive the keys required to decrypt objects.
We also require a security criterion.
Informally, the \emph{strong key-indistinguishability} criterion requires the following.
\begin{quote}
There is no polynomial time algorithm, given $z \in X$, a set of secrets $\sigma(Y) = \set{\sigma(y) : y \in Y}$ such that $z \not\leqslant y$ for any $y \in Y$, and $\kappa(x)$ for all $x \ne z$ (and the public information $Pub$), that can distinguish between $\kappa(z)$ and a random key in $\keysp$. 
\end{quote}
That is, an adversary cannot distinguish a key from random unless it may be computed from one of the secrets or keys known to the adversary (which implies, in particular, that the adversary can only compute such a key if it can be computed from one of those secrets); see Freire \etal~\cite{frei:simp13} for further details.

\subsection{Chain-based Enforcement}

For certain classes of cryptographic enforcement schemes, public information is not required.
In particular, if $X$ is a chain, then (by definition) there is a unique directed path from $x$ to $y$ (in the Hasse diagram of $X$) whenever $y < x$.
Then for $y \lessdot x$, we may define the \emph{secret} $\sigma(y)$ to be $F(\sigma(x))$, and $\kappa(y) = H(\sigma(y))$, where $F$ and $H$ are suitable one-way functions.
Thus, if $y < x$, there exist $z_1,\dots,z_\ell \in X$ with $y = z_1 \lessdot z_2 \lessdot \dots \lessdot z_\ell = x$; $\kappa(y)$ may be derived from $\sigma(x)$ by iteratively deriving $\sigma(z_i) = F(\sigma(z_{i+1}))$, $i = \ell-1,\dots,1$, and then deriving $\kappa(y) = H(\sigma(y)) = H(\sigma(z_1))$.

This observation has led to the development of chain-based CESs~\cite{CrDaMa10,frei:prov11,frei:simp13} for arbitrary information flow policies.
The basic idea is to partition the information flow policy $(X,\leqslant)$ into chains and then construct multiple CESs, one for each chain.

More formally, let $(X,\leqslant)$ be a poset and $C = x_1 > x_2 > \dots > x_m$ be a chain in $X$.
Then we say any chain of the form $x_j > x_{j+1} > \dots > x_m$, $1 \leqslant j \leqslant m$, is a \emph{suffix} of $C$; the empty chain is (vacuously) also a suffix of $C$.

\begin{proposition}
For all $x \in X$ and any chain $C \subseteq X$, $\dset{x} \cap C$ is a suffix of $C$.
\end{proposition}

The above result (due to Crampton~\etal~\cite[Proposition 4]{CrDaMa10}) enables us to define, for a given chain partition $\Pi$, the secrets that should be given to a user $u \in U_x$, since $\dset{x}$ defines the labels for which $u$ is authorized.
Given a chain partition $\Pi = \set{C_1,\dots,C_\ell}$, $\set{\dset{x} \cap C_1,\dots,\dset{x} \cap C_\ell}$ is a disjoint collection of chain suffixes.
Hence, a user in $U_x$ must be given the secrets for the maximal elements in the non-empty suffixes $\dset{x} \cap C_1,\dots,\dset{x} \cap C_\ell$.
Thus, any user requires at most $\ell$ secrets.
Let $\phi(x,\Pi) \subseteq X$ denote this set of maximal elements.
(Clearly, $x \in \phi(x,\Pi)$ for all chain partitions $\Pi$ and all $x \in X$.)

\begin{remark}
Let $w$ be the width of a poset $(X,\leqslant)$. 
Clearly, $(X,\leqslant)$ cannot have a chain partition with less than $w$ chains. 
Dilworth's theorem asserts that there exists a chain partition of $(X,\leqslant)$ into $w$ chains~\cite{dilw:deco50}.
Thus, if we can find a chain partition of $X$ into $w$ chains, no user will require more than $w$ secrets.
% Clearly, $u$ will not require more secrets than there are chains in the partition.
(If $u$ were to have more secrets than there are chains in the partition, then there must exist a chain containing $y$ and $z$ for which $u$ has secrets and one of the secrets may be derived from the other.)
% Moreover, any user assigned to the maximum element will require precisely $w$ secrets.
\end{remark}

Freire \etal~\cite{frei:simp13} provide a formal description of the $\setup$ and $\derive$ algorithms.
Informally, the $\setup$ algorithm performs the following steps:
\begin{enumerate}
 \item for each chain $C_i$  in $\Pi$, select a secret for the top element in $C_i$ and generate a secret for each element in the chain by applying the one-way function $F$ to the secret of its parent in $C_i$;
 \item for each element $x \in X$, generate $\kappa(x)$ by applying the one-way function $H$ to $\sigma(x)$;
 \item assign the secrets $\sigma(\phi(x,\Pi)) \defeq \set{\sigma(z) : z \in \phi(x,\Pi)}$ to each user in $U_x$.
\end{enumerate}
The $\derive$ algorithm performs the following steps, given $x,y \in X$ and $\sigma(\phi(x,\Pi))$:
\begin{enumerate}
 \item if $x = y$, then output $H(\sigma(x))$;
 \item if $y < x$, then find $z \in \phi(x,\Pi)$ such that $z \geqslant y$, so there exist $z = z_0 \gtrdot_{\Pi} \dots  \gtrdot_{\Pi} z_t = y$, and compute $F(\sigma(z_0)) = \sigma(z_1),\dots,F(\sigma(z_{t-1})) = \sigma(y)$; output $H(\sigma(y))$.
\end{enumerate}
This scheme has the strong key-indistinguishability property; see Freire \etal~\cite{frei:simp13} for further details.

A user in $U_x$ will need to be given $\card{\phi(x,\Pi)}$ secrets, in contrast to most CESs in the literature in which each user receives a single secret~\cite{atal:dyna09,cram:key06}.
However, chain-based CESs have substantial benefits:%
\begin{inparaenum}[(i)]
 \item they require no public information~\cite{CrDaMa10};\
 \item they can use cryptographic primitives that are very easy to compute; and 
 \item it is easy to construct schemes with the strong key-indistinguishability property~\cite{frei:simp13}.
\end{inparaenum}

\subsection{Problem Statement}\label{sec:problem-statement}

Certain aspects of chain-based CESs are not well understood.
As we have already noted, some users will require multiple secrets, each of which corresponds to a unique label in $X$.
In particular, a user $u$ in $U_x$ will require a secret for each chain that contains an element $y$ such that $y < x$.
Three chain partitions of the poset in Fig.~\ref{fig:hasse-diagram} are shown in Fig.~\ref{fig:chaincostcomparison}.
We have, for example, $\phi(g,\Pi_1) = \set{b,e,g}$, $\phi(g,\Pi_2) = \set{b,d,g}$, and $\phi(g,\Pi_3) = \set{d,g}$.
Hence, the number of secrets required, on a per-user basis and in total, will vary, depending on the chain partition chosen.
Thus, considering various chain partitions of $X$, we may ask:
\begin{itemize}
 \item How do we minimize $k_{\max}$, the maximum number of secrets a user may require?
 \item How do we minimize $K$, the total number of secrets required?
 \item How do we minimize $\widehat{K}$, the total number of secrets that need to be issued to users?
\end{itemize}

 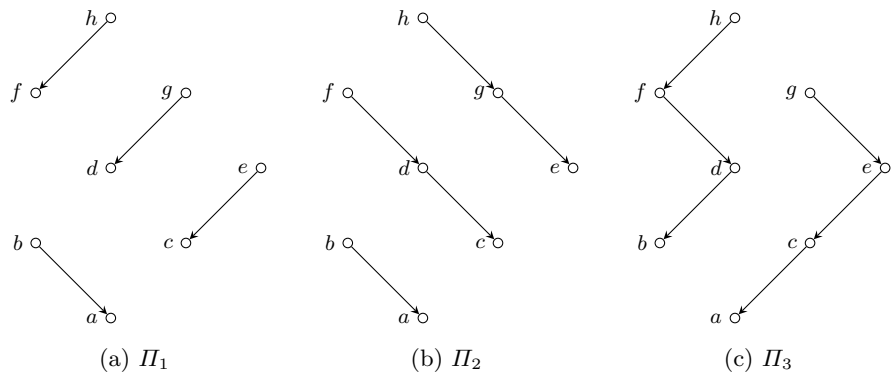
\begin{figure}[t]\centering
    \begin{subfigure}[t]{0.32\textwidth} \centering
    \begin{tikzpicture}[v/.style={circle,draw,fill=white,inner sep=0pt,minimum width=4pt},>=stealth,<-,scale=.9,transform shape]
      \node[v,label=left:$a$] (a) {};
      \node[v,above left=of a,label=left:$b$] (b) {};
      \node[v,above right=of a,label=left:$c$] (c) {};
      \node[v,above right=of b,label=left:$d$] (d) {};
      \node[v,above right=of c,label=left:$e$] (e) {};
      \node[v,above left=of d,label=left:$f$] (f) {};
      \node[v,above right=of d,label=left:$g$] (g) {};
      \node[v,above left=of g,label=left:$h$] (h) {};
      \draw (a) -- (b);
      \draw (c) -- (e);
      \draw (d) -- (g);
      \draw (f) -- (h);
    \end{tikzpicture} 
   \caption{$\Pi_1$}\label{subfig:chainpartition20keys}
   \end{subfigure}
   \hfill
    \begin{subfigure}[t]{0.32\textwidth} \centering
    \begin{tikzpicture}[v/.style={circle,draw,fill=white,inner sep=0pt,minimum width=4pt},>=stealth,<-,scale=.9,transform shape]
      \node[v,label=left:$a$] (a) {};
      \node[v,above left=of a,label=left:$b$] (b) {};
      \node[v,above right=of a,label=left:$c$] (c) {};
      \node[v,above right=of b,label=left:$d$] (d) {};
      \node[v,above right=of c,label=left:$e$] (e) {};
      \node[v,above left=of d,label=left:$f$] (f) {};
      \node[v,above right=of d,label=left:$g$] (g) {};
      \node[v,above left=of g,label=left:$h$] (h) {};
      \draw (a) -- (b);
%       \draw (a) -- (c);
%       \draw (b) -- (d);
      \draw (c) -- (d);
%       \draw (c) -- (e);
      \draw (d) -- (f);
      %\draw (d) -- (g);
      \draw (e) -- (g);
%       \draw (f) -- (h);
      \draw (g) -- (h);
    \end{tikzpicture} 
   \caption{$\Pi_2$}\label{subfig:chainpartition17keys}
   \end{subfigure}
    \hfill
    \begin{subfigure}[t]{0.32\textwidth} \centering
    \begin{tikzpicture}[v/.style={circle,draw,fill=white,inner sep=0pt,minimum width=4pt},>=stealth,<-,scale=.9,transform shape]
      \node[v,label=left:$a$] (a) {};
      \node[v,above left=of a,label=left:$b$] (b) {};
      \node[v,above right=of a,label=left:$c$] (c) {};
      \node[v,above right=of b,label=left:$d$] (d) {};
      \node[v,above right=of c,label=left:$e$] (e) {};
      \node[v,above left=of d,label=left:$f$] (f) {};
      \node[v,above right=of d,label=left:$g$] (g) {};
      \node[v,above left=of g,label=left:$h$] (h) {};
      %\draw (a) -- (b);
      \draw (a) -- (c);
      \draw (b) -- (d);
      %\draw (c) -- (d);
      \draw (c) -- (e);
      \draw (d) -- (f);
      %\draw (d) -- (g);
      \draw (e) -- (g);
      \draw (f) -- (h);
      %\draw (g) -- (h);
    \end{tikzpicture} 
   \caption{$\Pi_3$}\label{subfig:chainpartition13keys}
   \end{subfigure}
   \caption{Three chain partitions of the poset in Fig.~\ref{fig:hasse-diagram}} \label{fig:chaincostcomparison} 
\end{figure}

More formally, given a chain partition $\Pi$ of $(X,\leqslant)$, we may regard $\phi$ as a function from $X$ to $2^X$ that is completely determined by $\Pi$.
Thus, given a chain partition $\Pi$, we can define the following values.
\begin{align*}
 k_{\max}(\Pi) &\defeq \max\set{\card{\phi(x,\Pi)} : x \in X} \\  
 K(\Pi) &\defeq \sum_{x \in X} \card{\phi(x,\Pi)} \\
 \widehat{K}(\Pi) &\defeq \sum_{x \in X} \card{U_x} \cdot \card{\phi(x,\Pi)}
\end{align*}
Values of $k_{\max}$ and $K$ for the chain partitions in Fig.~\ref{fig:chaincostcomparison} are shown in Table~\ref{tbl:comparison-of-k-K}; node $h$ is used for illustrative purposes.\footnote{Note that we can deduce  $K$ from $ \widehat{K}$ by letting $|U_x|=1$ for all $x \in X$.}

\begin{table}[h]\centering
%     \begin{subfigure}{\textwidth}
    \[
     \begin{array}{c|rrrr}
      \text{Partition} & \quad\phi(h) & \quad k_{\max} & \quad K \\
     \hline
     \Pi_1 & \set{b,e,g,h} & 4 & 20 \\
     \Pi_2 & \set{b,f,h} & 3 & 17 \\
     \Pi_3 & \set{g,h} & 2 & 13  \\
     \end{array}
    \]
    \caption{$\phi(h)$, $k_{\max}$ and $K$ for the chain partitions in Fig.~\ref{fig:chaincostcomparison}}\label{tbl:comparison-of-k-K}
%    \end{subfigure}
% 
\end{table}
The important question is: Can we minimize these parameters (over all choices of chain partition $\Pi$ for $X$)?
In short, given an information flow policy $(X,\leqslant)$, how do we determine $\Pi$ for use in a chain-based CES?\footnote{Crampton~\etal~\cite{CrDaMa10} observed that further research was needed to identify the best choice of chain partition for a given information flow policy.  While subsequent research has formalized~\cite{frei:prov11} and strengthened the security properties of chain-based CESs~\cite{frei:simp13}, we are not aware of any research that specifies how to select a chain partition.}
It is this question we address in the remainder of the paper.
% In particular, we show it is possible, given an arbitrary information flow policy $(X,\leqslant)$ to find a chain partition $\Pi$ such that the corresponding CES has the following properties:
% \begin{itemize}
%  \item no user requires more than $w$ secrets, where $w$ is the width of $X$;
%  \item $K(\Pi)$ and $\widehat{K}(\Pi)$ are minimal with respect to all possible chain partitions.\footnote{It is easy to adapt this so that we minimize $K$; specifically, we assume $\card{U_x} = 1$ for all $x \in X$.}
% \end{itemize}
% Moreover, our construction requires time polynomial in the size of $X$.
In particular, at the end of Section~\ref{sec:chain-partition-requiring-widehat-k-keys}, we prove the following result.

\begin{theorem}\label{thm:main-theorem}
 Let $(X,\leqslant)$ be an information flow policy of width $w$ and let $\widehat{K}$ denote the minimum number of secrets required by a chain-based enforcement scheme for $X$.
 Then in $O(|X|^4w)$  time, we can find a chain partition $\Pi$ for which the corresponding chain-based enforcement scheme only requires $\widehat{K}$ secrets and $k_{\max} \leqslant w$.
%  no user requires more than $w$ secrets, where $w$ is the width of $(X, \leqslant)$.
\end{theorem}

\begin{remark}
We assume throughout that our information flow policy has a maximum element.
We may assume this without loss of generality: given an information flow policy $(X,\leqslant)$ without a maximum element, we simply add a maximum element $r$ and define $r \gtrdot m$ for all maximal elements $m$ in $X$; no users are assigned to $r$. 
Observe that such a transformation does not affect the values of $k_{\max}$ and $\widehat{K}$.
\end{remark}

\section{Computing $k_{\max}$ and $\widehat{K}$}\label{sec:computing-k-max-and-widehat-k}

Informally, we take a poset $(X,\leqslant)$ and construct a second poset $(X,\leqslant')$, where $x <' y$ implies $x < y$ (but $x < y$ does not necessarily imply $x <' y$).
We will say $\leqslant'$ is \emph{contained} in $\leqslant$.
In particular, any chain partition $\Pi$ of $(X,\leqslant)$ defines a second poset $(X,\leqslant_\Pi)$, where $x <_\Pi y$ if and only if $x$ and $y$ belong to the same chain and $x < y$; thus $\leqslant_\Pi$ is contained in $\leqslant$ for any $\Pi$.
Note, however, that $x \lessdot_\Pi y$ does not necessarily imply $x \lessdot y$.%
\footnote{To see this, consider the poset of four elements, in which $a \lessdot b \lessdot d$ and $a \lessdot c \lessdot d$ with $b \not\leqslant c, c \not\leqslant b$.  Then $\set{\set{b},\set{c},\set{a,d}}$ is a chain partition and $a \lessdot_\Pi d$, but $a \nlessdot d$.}

Given a poset $(X,\leqslant)$ and $z < y$, we define
\[
 \gamma(yz) = \set{x \in X : x \geqslant z, x \not\geqslant y}.
\]
Thus $z \in \gamma(yz)$ and $y \not\in \gamma(yz)$.
For the maximum element $r \in X$ and any $y,z \in X$ such that \mbox{$z < y$}, $r \not\in \gamma(yz)$.
Informally, the intuition behind $\gamma$ is that its cardinality measures the ``damage'' that would be done by creating a chain partition $\Pi$ such that $z \lessdot_{\Pi} y$, because having $z \lessdot_\Pi y$ means that $z \not\leqslant_\Pi x$ for any $x \in \gamma(yz)$.
Thus, every user in $U_x$ will require an extra secret in order to derive $\kappa(z)$.
We will capture this intuition more precisely in Lemma~\ref{lem:z-in-phi-x-and-x-in-gamma-y-z}.

\begin{remark}\label{rem:constructing-tree-from-chain-partition}
For maximum element $r$ and any chain partition $\Pi = \set{C_1,\dots,C_\ell}$,  $\phi(r,\Pi) = \set{t_1,\dots,t_\ell}$, where $t_i$ is the maximum element in chain $C_i$.
Moreover, $r = t_i$ for some $i$.
Hence, we can construct a tree $\treepi = (X,\leqslant_{\treepi})$, where $y \lessdot_{\treepi} x$ if and only if one of the following conditions holds:%
\begin{inparaenum}[(i)]
 \item $y = t_j$, $j \ne i$, and $x = r$;
 \item $y \lessdot_{\Pi} x$.
\end{inparaenum}
\end{remark}

Figure~\ref{fig:out-tree-from-chain-partition} illustrates the construction of two such trees, using chain partitions from Fig.~\ref{fig:chaincostcomparison}; the arcs used to create the trees are shown as dashed lines.

\begin{figure}[h]\centering
    \begin{subfigure}[t]{0.45\textwidth} \centering
    \begin{tikzpicture}[v/.style={circle,draw,fill=white,inner sep=0pt,minimum width=4pt},>=stealth,<-]
      \node[v,label=left:$a$] (a) {};
      \node[v,above left=of a,label=left:$b$] (b) {};
      \node[v,above right=of a,label=left:$c$] (c) {};
      \node[v,above right=of b,label=left:$d$] (d) {};
      \node[v,above right=of c,label=left:$e$] (e) {};
      \node[v,above left=of d,label=left:$f$] (f) {};
      \node[v,above right=of d,label=left:$g$] (g) {};
      \node[v,above left=of g,label=left:$h$] (h) {};
      \draw (a) -- (b);
      %\draw (a) -- (c);
      %\draw (c) -- (d);
      \draw (c) -- (e);
      %\draw (d) -- (f);
      \draw (d) -- (g);
      \draw[thick,dashed] (b) -- (h);
      \draw[thick,dashed] (e) to[bend right] (h);
      \draw (f) -- (h);
      \draw[thick,dashed] (g) -- (h);
    \end{tikzpicture} 
   \caption{$\treepi_1$}\label{subfig:pitree-20keys}
   \end{subfigure}
   \hfill
   \begin{subfigure}[t]{0.45\textwidth} \centering
    \begin{tikzpicture}[v/.style={circle,draw,fill=white,inner sep=0pt,minimum width=4pt},>=stealth,<-]
      \node[v,label=left:$a$] (a) {};
      \node[v,above left=of a,label=left:$b$] (b) {};
      \node[v,above right=of a,label=left:$c$] (c) {};
      \node[v,above right=of b,label=left:$d$] (d) {};
      \node[v,above right=of c,label=left:$e$] (e) {};
      \node[v,above left=of d,label=left:$f$] (f) {};
      \node[v,above right=of d,label=left:$g$] (g) {};
      \node[v,above left=of g,label=left:$h$] (h) {};
%       \draw (a) -- (b);
      \draw (a) -- (c);
      \draw (b) -- (d);
      %\draw (c) -- (d);
      \draw (c) -- (e);
      \draw (d) -- (f);
      %\draw (d) -- (g);
      \draw (e) -- (g);
      \draw (f) -- (h);
      \draw[thick,dashed] (g) -- (h);
    \end{tikzpicture} 
   \caption{$\treepi_3$}\label{subfig:pitree-13keys}
  \end{subfigure}
  \caption{Creating trees from partitions $\Pi_1$ and $\Pi_3$ in Fig.~\ref{fig:chaincostcomparison}}\label{fig:out-tree-from-chain-partition}
\end{figure}
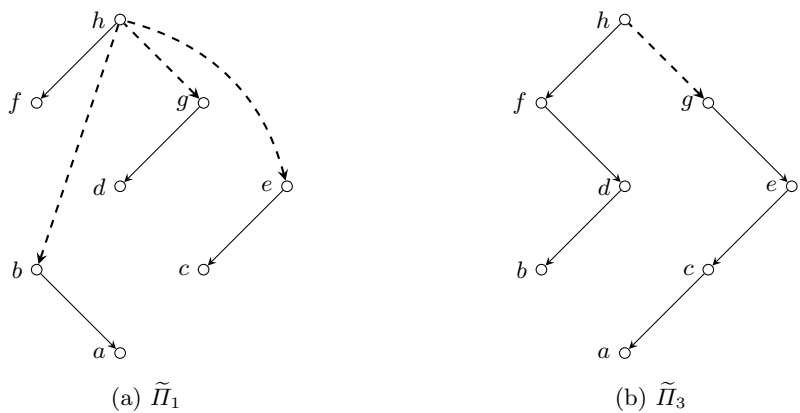

\begin{lemma}\label{lem:z-in-phi-x-and-x-in-gamma-y-z}
 Let $(X,\leqslant)$ be a poset and let $\Pi$ be a chain partition of $X$.
 Then, for all $x,y,z \in X$ such that $x \ne r$ and $z \lessdot_{\treepi} y$,
 \[
  z \in \phi(x,\Pi)\ \text{if and only if}\ x \in \gamma(yz).
 \]
\end{lemma}

\begin{proof}
  Given $z \in \phi(x,\Pi)$ and chain partition $\Pi = \set{C_1,\dots,C_\ell}$, \mbox{$y_i \in \phi(x,\Pi) \cap C_i$} if and only if $C_i \cap \dset{x}$ is non-empty and $y_i$ is the maximum element in $C_i \cap \dset{x}$ (Sec.~\ref{sec:problem-statement}).
  Thus, $z \leqslant x$.
  Moreover, $x \not\geqslant y$ (otherwise there would exist $t \in \phi(x)$ such that $y \leqslant_{\treepi} t$ and hence $z \lessdot_{\treepi} y \leqslant_{\treepi} t$, violating the condition that $z$ is the maximum element in the suffix $C_i \cap \dset{x}$).
  That is, $x \in \gamma(yz)$.
  
  Now suppose $x \in \gamma(yz)$.
  Then $x \not\geqslant y$, by definition, and hence $y$ does not belong to $\dset{x} \cap C_i$ for any $i$.
  However, $x \geqslant z$; hence, there exists $t \in \phi(x)$ such that $z \leqslant_{\treepi} t$.
  Since $\Pi$ is a chain partition, the only parent of $z$ in $\treepi$ is $y$.
  Hence it must be the case that $z = t$ (and thus $z \in \phi(x)$).\qed
\end{proof}

Let $(X,\leqslant)$ be an information flow policy and let $y,z \in X$ with $z < y$.
Then, following Crampton~\etal~\cite{CrFaGuJoPo14}, we define 
\[
 \omega(yz) \defeq \sum_{x \in \gamma(yz)} \card{U_x}.
\]
We will be interested in minimizing the $\sum \omega(yz)$, where the sum is taken over all pairs $(y,z)$ such that $z \lessdot_{\treepi} y$.
The intuition behind this definition is that it captures, in some appropriate sense, the connectivity that is lost from $(X,\leqslant)$ by using $(X,\leqslant_{\Pi})$.
Since every element in $(X,\leqslant_\Pi)$ has at most one parent, $\gamma(yz)$ represents those elements in $X$ that become ``disconnected'' from $z$ by defining $z \lessdot_\Pi y$.
% Thus the cardinality of $\gamma(\treepi,yz)$ is one way of measuring how many additional secrets will be required for $z$.
The next result establishes an exact correspondence between $\phi(x,\Pi)$ and $\gamma(yz)$, and enables us to use network flow techniques to compute a chain partition that minimizes $\widehat{K}$ (as we explain in Sec.~\ref{sec:chain-partition-requiring-widehat-k-keys}).

\begin{theorem}\label{thm:keys-equal-sum-of-omega}
Let $(X,\leqslant_\Pi)$ be a chain partition of $(X,\leqslant)$ with maximum element~$r$.
Then
\[
 \widehat{K}(\Pi) = \ell \card{U_r} + \sum_{z \lessdot_{\treepi} y} \omega(yz)
\]
where $\ell$ is the number of chains in $\Pi$.
\end{theorem}

\begin{proof}
By definition,
\[
 \widehat{K}(\Pi) = \sum_{x \in X} \card{U_x} \card{\phi(x,\Pi)}
 =  \card{U_r} \card{\phi(r,\Pi)} + \sum_{x \in X\setminus r} \sum_{z \in X} \card{U_x} \delta(x,z), % \sum_{\stackrel{x \in X}{z \in \phi(x,\Pi)}} \card{U_x}. 
\]
where $\delta(x,z)$ equals $1$ if $z \in \phi(x, \Pi)$ and $0$ otherwise.
By Lemma~\ref{lem:z-in-phi-x-and-x-in-gamma-y-z}, we have $\delta(x,z) = 1$ if and only if $x \in \gamma(yz)$ for $z \lessdot_{\treepi} y$.
Moreover, $y$ is unique, since $\treepi$ is a tree.
Therefore, 
\[
 \sum_{x \in X\setminus r} \sum_{z \in X} \card{U_x} \delta(x,z) = \sum_{z \lessdot_{\treepi} y} \sum_{x \in \gamma(yz)} \card{U_x} = \sum_{z \lessdot_{\treepi} y} \omega(yz)
\]
As $r \geqslant x$ for all $x \in X$, $\phi(r,\Pi)$ must contain exactly one element from each chain in $\Pi$. Therefore $\card{U_r} \card{\phi(r,\Pi)} = \ell \card{U_r}$, as required.
\qed
\end{proof}

The following result shows that the number of secrets required by a chain partition can be computed by considering only the minimum elements in the chain partition.

\begin{lemma}\label{lem:number-of-keys-from-bottom-elements}
 Let $\Pi = \set{C_1,\dots,C_\ell}$ be a chain partition of $(X,\leqslant)$ and let chain $C_i$ have bottom element $b_i$, $1 \leqslant i \leqslant \ell$.
 Then 
 \[
  K(\Pi) = \sum_{i=1}^\ell \card{\uset{b_i}}\qquad\text{and}\qquad \widehat{K}(\Pi) = \sum_{i=1}^\ell \sum_{x \in \uset{b_i}}\card{U_x}.
 \]
\end{lemma}

\begin{proof}
 We have, by definition,
 \begin{align*}
  \widehat{K}(\Pi) &= \sum_{x \in X}  \card{U_x} \card{\phi(x,\Pi)} 
         = \sum_{x \in X}  \card{U_x} \card{\set{C_i : C_i \cap \dset{x} \ne \emptyset, 1 \leqslant i \leqslant \ell}} \\
         &= \sum_{x \in X}  \card{U_x} \card{\set{b_i : x \geqslant b_i,1 \leqslant i \leqslant \ell}} \\
         &= \sum_{x \in X} \sum_{i=1}^\ell  \card{U_x} \delta(x,b_i) \qquad\text{where $\delta(x,b_i) = 1$ if $x \geqslant b_i$ and $0$ otherwise} \\
         &= \sum_{i=1}^\ell \sum_{x \in X}  \card{U_x} \delta(x,b_i) 
         = \sum_{i=1}^\ell \sum_{x \in \uset{b_i}}\card{U_x}
%          = \sum_{i=1}^\ell  \card{U_x} \card{\set{x \in X : x \geqslant b_i}} 
%          &= \sum_{x \in X} \sum_{i=1}^\ell \card{\set{b_i : 
%          &= \sum_{x \in X,\,b_i \leqslant x} 1 \\
%          = \sum_{i=1}^\ell  \card{U_x} \card{\uset{b_i}}
 \end{align*}
Clearly, we may prove the result for $K$ in an analogous fashion.\qed
\end{proof}

In Fig.~\ref{subfig:chainpartition20keys}, for example, the bottom elements are $a$, $c$, $d$ and $f$ and $\card{\uset{a}} = 8$, $\card{\uset{c}} = 6$, $\card{\uset{d}} = 4$ and $\card{\uset{f}} = 2$.
Thus, the number of secrets required in total is $20$.

\begin{theorem}\label{thm:chain-partition-only-requires-w-chains}
 Let $(X,\leqslant)$ be an information flow policy of width $w$ and let $\widehat{K}$ denote the minimum number of secrets required by a chain-based enforcement scheme for $X$.
 Then there exists a chain partition containing $w$ chains 
 such that $\widehat{K}(\Pi) = \widehat{K}$.
%  for which the corresponding chain-based enforcement scheme only requires $\widehat{K}$ secrets, where $w$ is the width of $(X,\leqslant)$.
\end{theorem}

\begin{proof}
 Let $\Pi$ be a chain partition of $X$ into $t \geqslant w$ chains such that $\widehat{K}(\Pi) = \widehat{K}$ and let $B$ be the set of bottom vertices in the chains of $\Pi$.
 A result of Gallai and Milgram asserts that if a chain partition $\Pi$ of a poset $(X,\leqslant)$ contains $t$ chains, where $t > w$, then there exists a chain partition $\Pi'$ into $t-1$ chains such that the set of bottom vertices in $\Pi'$ is a subset of $B$~\cite{GaMi60}.\footnote{The result is phrased in the language of digraphs, but every poset may be represented by an equivalent transitive acyclic digraph.}
 Hence, by iterated applications of the Gallai-Milgram result, there exists a chain partition $\Pi^*$ of width $w$ such that the set of bottom vertices $B^*$ in $\Pi^*$ is a subset of $B$.
 Moreover, by Lemma~\ref{lem:number-of-keys-from-bottom-elements},
 \[
  \widehat{K}(\Pi^*) = \sum_{b \in B^*}\sum_{x \in \uset{b}}\card{U_x}  \leqslant  \sum_{b \in B} \sum_{x \in \uset{b}}\card{U_x}
%   \sum_{b \in B^*} \card{U_x} \card{\uset{b}} \leqslant \sum_{b \in B} \card{U_x} \card{\uset{b}} = \widehat{K}
 \]
 By the minimality of $\widehat{K}$, we deduce that $\widehat{K}(\Pi^*) = \widehat{K}$.\qed
\end{proof}

\begin{corollary}
Let $(X,\leqslant)$ be an information flow policy.
There exists a chain partition such that the total number of secrets $\widehat{K}$ is minimized and $k_{\max} \leqslant w$.
% no user requires more than $w$ secrets.
\end{corollary}

\begin{proof}
 The result follows immediately from Theorem~\ref{thm:chain-partition-only-requires-w-chains}, the definition of $k_{\max} = \max\set{\card{\phi(x,\Pi)} : x \in X}$, and the fact that $\card{\phi(x,\Pi)}$ is bounded above by the number of chains in $\Pi$ for all $x \in X$.\qed
\end{proof}

% \section{Chain Partitions Containing $w$ Chains}\label{sec:partitions-containing-w-chains}
\section{Finding a Chain Partition Requiring $\widehat{K}$ Keys}\label{sec:chain-partition-requiring-widehat-k-keys}

Suppose $(X,\leqslant)$ is a poset of width $w$.
In general, a chain partition of $X$ has $\ell \geqslant w$ chains.
Theorem~\ref{thm:chain-partition-only-requires-w-chains} asserts that there exists a partition of $X$ into $w$ chains such that the corresponding enforcement scheme requires the minimum number of secrets.
We now show how such a chain partition may be constructed.
In particular, we show how to transform the problem of finding a chain partition $\Pi$ such that $\widehat{K}(\Pi)$ attains the minimum value into a problem of finding a minimum cost flow in a network.

Informally, a \emph{network} is a directed graph in which each edge is associated with a \emph{capacity}.
A \emph{network flow} associates each edge in a given network with a flow, which must not exceed the capacity of the edge.
Networks are widely used to model systems in which some quantity passes through channels (edges in the network) that meet at junctions (vertices); examples include traffic in a road system, fluids in pipes, or electrical current in circuits.
In our setting, we model an information flow policy as a network in which the capacities are determined by the weights $\omega$.
Our definitions for networks and network flows follow the presentation of Bang-Jensen and Gutin~\cite{BaGu02}.

\begin{definition}
 A \emph{network} is a tuple $\mathcal{N} = (D,l,u,c,b)$, where:
 \begin{itemize}
  \item $D = (V,A)$ is a directed graph with vertex set $V$ and arc set $A$;
  \item $l : V \times V \rightarrow \mathbb{N}$ such that $l(vv') = 0$ if $vv' \not\in A$ and $l(vv') \geqslant 0$ otherwise;
  \item $u : V \times V \rightarrow \mathbb{N}$ such that $u(vv') = 0$ if $vv' \not\in A$ and $u(vv') \geqslant l(vv') \geqslant 0$ otherwise;
  \item $c : V \times V \rightarrow \mathbb{R}$;
  \item $b : V \rightarrow \mathbb{R}$ such that $\sum_{v \in V} b(v) = 0$.
 \end{itemize}
\end{definition}
% 
% \begin{definition}
% A \emph{network} is a directed graph $D=(V,A)$ (where $V$ is the vertex set and $A \subset V \times V$ is the arc set) associated with the following values for each $(i,j) \in V \times V$: 
% \begin{itemize}
%  \item a \emph{lower bound} $l_{ij} \geqslant 0$; 
%  \item an \emph{upper bound} $u_{ij} \geqslant l_{ij}$;
%  \item a \emph{cost} $c_{ij}$
% \end{itemize}
% For every arc $ij \in V \times V$, if $ij \not\in A$, $l_{ij} = u_{ij} = 0$.
% In addition, for each $i \in V$ there is an integer value $b_i$, such that $\sum_{i \in V}b_i = 0$.
% \end{definition}
% 
% We may interpret $l$ and $u$ as functions from $V \times V$ into the set of non-negative integers; $c$ as a function from $V \times V$ into the set of integers; and $b$ as a function from $V$ into the set of integers.

Intuitively, $l$ and $u$ represent lower and upper bounds, respectively, on how much flow can pass through each arc, and $c$ represents the cost associated with each unit of flow in each arc. The function $b$ represents how much flow should enter or leave the network at a given vertex. If $b(x)=0$, then the flow going into $x$ should be equal to the flow going out of $x$. If $b(x)>0$, then there should be $b(x)$ more flow coming out of $x$ than going into $x$. If $b(x)<0$, there should be $|b(x)|$ more flow going into $x$ than coming out of $x$.
% [Perhaps some intuition behind the various functions in the definition of a network would be good\dots?]

\begin{definition}
Given a network $\mathcal{N} = (D, l,u, c, b)$, a function $f:V \rightarrow \mathbb{N}$ is a \emph{feasible flow} for $\mathcal{N}$ if the following conditions are satisfied:
 \begin{itemize}
  \item $u(vv') \geqslant f(vv') \geqslant l(vv')$ for every $vv' \in V \times V$;
  \item $\sum_{v' \in V}(f(vv') - f(v'v)) = b(v)$ for every $v \in V$.
 \end{itemize}
The \emph{cost} of $f$ is defined to be \[ \sum_{vv' \in A} c(vv') f(vv'). \]
\end{definition}

% Note that $f_{ij} = 0$ in any feasible flow when $ij \not\in A$.
% Thus the cost of a feasible flow is $\sum_{ij \in A} c_{ij} f_{ij}$.

% Let $(X,\leqslant)$ be an information flow policy and let $H = (X,\eofh)$ be the Hasse diagram of $X$.
% Then $H^*$ is a transitive acyclic diagraph with a unique vertex of in-degree $0$.
Our aim is to find a tree $\treepi$ such that $\Pi$ is a chain partition of $X$ with $w$ chains that minimizes $\widehat{K}$.
To do this, we will construct a network $\mathcal{N}$ such that the minimum cost flow of $\mathcal{N}$ corresponds to the desired tree.
We can then find the minimum cost flow of $\mathcal{N}$ in polynomial time.

In $\treepi$, we want every vertex except $r$ to have at most one parent and at most one child.
We cannot represent this requirement directly in a network.
However, we can use the \emph{vertex splitting procedure}~\cite{BaGu02} to simulate it.
Specifically, given poset $(X,\leqslant)$, define $\vin[X] = \set{\vin : x \in X\setminus \{r\}}$ and $\vout[X] = \set{\vout : x \in X}$; and define $v' \prec v$ if and only if either $v = \vin$ and $v'=\vout$ for some $x \in X\setminus r$, or $v = \vout[x]$ and $v' = \vin[y]$ for some $x,y \in X$ such that $y < x$.
We now add a minimum element $\bot$, where $\bot \prec\vout$ for all $x \in X$.
% Then $(\vin[X] \cup \vout[X] \cup \set{\bot}, \preccurlyeq)$ is a poset.

Then define $D = (\vin[X] \cup \vout[X] \cup \set{\bot},A)$, where $xy \in A$ if and only if $y \prec x$, and the network $(D,l,u,c,b)$, where
\begin{align*}
 l(vv') &= %
   \begin{cases} 
    1 & \text{if $v=\vin, v'=\vout, x \in X\setminus r$} \\
    0 & \text{otherwise;}
   \end{cases} \\
 u(vv') &= 
   \begin{cases}
    1 & \text{if $v' \prec v$} \\
    0 & \text{otherwise;}
   \end{cases} \\
 c(vv') &= %
  \begin{cases}
    \omega(xy) & \text{if  $v = \vout, v' = \vin[y], y \leqslant x$} \\
    0  & \text{otherwise;} \\
  \end{cases} \\
 b(v) &= 
 \begin{cases}
  w & \text{if $v = \vout[r]$} \\
  -w & \text{if $v = \bot$} \\
  0 & \text{otherwise.}
 \end{cases}
\end{align*}
We call this network the \emph{network chain-representation of $(X,\leqslant)$}.
Note that any feasible flow $f$ for this network must have $0 \leqslant f(xy) \leqslant 1$ for all $xy \in A$.

\begin{lemma}\label{lem:min-cost-flow-equals-min-keys}
 Let $\mathcal{N}$ be the network chain-representation of poset $(X,\leqslant)$.
 Then the minimum number of secrets required by a chain-based enforcement scheme for $(X,\leqslant)$ with $w$ chains is $w \card{U_r} + \widehat{f}$, where $\widehat{f}$ is the minimum cost of a feasible flow in $\mathcal{N}$.
\end{lemma}

\begin{proof}
 Suppose we are given a chain partition $\Pi$ with $w$ chains.
 Then we may construct the tree $\treepi$.
 Consider the following flow:
 \begin{align*}
  f(\vin\vout) &= 1\qquad \text{for all $x \in X\setminus r$}; \\
  f(\vout\vin[y]) &= 1\qquad \text{if $y \lessdot_{\treepi} x$}; \\
  f(\vout\bot) &= 1\qquad \text{if $x$ is a bottom element in a chain in $\Pi$}; \\
  f &= 0 \qquad \text{otherwise}.
 \end{align*}
 Then we can show that $f$ is a feasible flow.
Indeed, by construction all arcs $xy$ satisfy  $u(xy) \geqslant f(xy) \geqslant l(xy)$. In the graph formed by arcs $xy$ with $f(xy)=1$, it is clear that every vertex $x$ has in-degree and out-degree $1$, except for $\vout[r]$ and $\bot$. As there is one element $y$ such that $y \lessdot_{\treepi} r$ for each chain in $\Pi$, $\vout[r]$ has in-degree $0$ and out-degree $w$ in this graph, and similarly $\bot$ has in-degree $w$ and out-degree $0$.
As all arcs $xy$ have $f(xy)=1$ or $f(xy)=0$, we have that \[ \sum_{v \in V(D)}(f(xv) - f(vx)) = b(x) \] for all $x$, as required.
Moreover, the cost of $f$ equals $\sum_{x \lessdot_{\treepi} y} \omega(yx)$.
 
Conversely, suppose $f$ is a feasible flow for $\mathcal{N}$.
Then we define $y \lessdot_f x$ if and only if $f(\vout,\vin[y]) = 1$. 
For each $x \in X\setminus r$, the arc $\vin\vout$ is the only in-coming arc for $\vout$ and the only out-going arc for $\vin$ in $D$, and by definition of $\mathcal{N}$, $f(\vin\vout) = 1$. As $b(\vin)=b(\vout)=0$ and all in-coming arcs for $\vin$ are of the form $\vout[y]\vin$, it follows that there is exactly one element $y \in X$ such that $x \lessdot_f y$, and at most one element $z \in X$ such that $z \lessdot_f x$.
As $b(\vout[r])=w$ and $\vout[r]$ has no in-coming arcs in $D$, and all its out-going arcs are of the form $\vout[r]\vin$, there are exactly $w$ elements $y$ such that $y \lessdot_f r$. Let these elements be labelled $t_1, \dots, t_w$.

Now choose an arbitrary $i$, $1 \leqslant i \leqslant w$, and define $y \lessdot_{\Pi} x$ if and only if $x=r$ and $y=t_i$, or $x \neq r$ and $y \lessdot_f x$.
Then for every element $x \in X$, there is at most one element $y \in X$ such that $x \lessdot_{\Pi} y$, and at most one element $z \in X$ such that $z \lessdot_{\Pi} x$.

It is easy to see that $\leqslant_{\Pi}$, the reflexive, transitive closure of $\lessdot_{\Pi}$, defines a chain partition of $X$. (Observe that as $D$ is an acyclic digraph, the transitive reflexive closure of $\lessdot_{\Pi}$ is antisymmetric, and therefore a partial order. The fact that $(X, \leqslant_{\Pi})$ is a chain partition can be shown by induction on $|X|$, considering $X$ with a minimal element removed for the induction step.)
By construction, the only maximal elements for $\leqslant_{\Pi}$ are $r$ and the elements $t_j$ for $j \neq i$. Thus,  $(X, \leqslant_{\Pi})$ has $w$ chains.

Recall the definition of $\lessdot_{\treepi}$, that  $y \lessdot_{\treepi} x$ if and only if either $y \lessdot_{\Pi} x$, or $y = t_j$, $j \ne i$, and $x = r$. Note that  $\lessdot_{\treepi}$ is exactly the relation $\lessdot_f$. 
By Theorem \ref{thm:keys-equal-sum-of-omega}, the number of secrets required by $\Pi$ is 
 \[ 
   w \card{U_r} + \sum_{z \lessdot_{\treepi} y} \omega(yz). 
 \]
As $z \lessdot_{\treepi} y$ if and only if $f(\vout[y]\vin[z]) = 1$, $c(\vout[y]\vin[z]) = \omega(yz)$, and $c(uv)=0$ for all other arcs with $f(uv)=1$, we have that $\sum_{z \lessdot_{\treepi} y} \omega(yz)$ is exactly the cost of $f$, as required.\qed
\end{proof}

\begin{lemma}\label{lem:min-cost-flow-in-poly-time}
 We can find a minimum cost flow for $\mathcal{N}$ in $O(|X|^4w)$ time.
\end{lemma}
\begin{proof}
The Negative Cycle algorithm (see~\cite[\S5.3]{Ahuja93}, for example) finds a minimum cost flow for a network with $n$ vertices and $m$ arcs in time $O(nm^2CU)$, where $C$ denotes the maximum cost on an arc, and $U$ denotes the maximum of all upper bounds on arcs and the absolute values of all balance demands on vertices. 
By construction of $\mathcal{N}$, we have that $n = 2|X| = O(|X|)$, $m = O(n^2)=O(|X|^2)$, $C = \max\set{\omega(xy): xy \in \eofg} = O(|X|)$, $U = 1$ and $C = w$.
Thus we get the desired running time.\qed
\end{proof}

\begin{remark}
Strictly speaking, the Negative Cycle algorithm assumes that all lower bounds on arcs are $0$. 
% But this can be enforced by adjusting the network as follows: for any arc $xy$ with $l(xy)>0$, replace $\mathcal{N}$ with $\mathcal{N}' = (D,l',u',c,b')$, where  $l'(xy)=0$,  $u'(xy) = u(xy) - l(xy)$, $b'(x) = b(x) - l(xy)$ and $b'(y) = b(y) + l(xy)$, and $l'=l$, $u'=u$ and $b=b'$ otherwise.
However, we can satisfy this assumption, given $\mathcal{N} = (D,l,u,c,b)$, by defining the network $\mathcal{N}' = (D,l',u',c,b')$, where
\begin{alignat*}{2}
 l'(xy) &= 0 & \qquad & b'(x) = b(x) - l(xy) \\
 u'(xy) &= u(xy) - l(xy) & \qquad & b'(y) = b(y) + l(xy) 
\end{alignat*}
Then the minimum cost flow $f'$ for $\mathcal{N}'$ will have cost exactly $\sum_{xy} l(xy)c(xy)$ less than the minimum cost flow for $\mathcal{N}$, and $f'$ can be transformed into a minimum cost feasible flow $f$ for $\mathcal{N}$ by setting $f(xy) = f'(xy)+l(xy)$.% and $f=f'$ otherwise.
\end{remark}

We are now able to prove our main result, which is, essentially, a corollary of Theorem~\ref{thm:chain-partition-only-requires-w-chains} and Lemmas~\ref{lem:min-cost-flow-equals-min-keys} and~\ref{lem:min-cost-flow-in-poly-time}.

\begin{proof}[of Theorem~\ref{thm:main-theorem}]
By Theorem \ref{thm:chain-partition-only-requires-w-chains}, there exists a chain partition  that has exactly $w$ chains, for which the corresponding chain-based enforcement scheme only requires $\widehat{K}$ secrets.
Then by Lemma \ref{lem:min-cost-flow-equals-min-keys},  $\widehat{K}$  is equal to the minimum cost of a feasible flow in $\mathcal{N}$, the network chain-representation of $(X, \leqslant)$. 
By Lemma \ref{lem:min-cost-flow-in-poly-time}, such a flow can be found in $O(|X|^4w)$ time, and this flow can be easily transformed into the corresponding chain partition $\Pi$. 
Finally, by definition of $\phi(x, \Pi)$, $|\phi(x, \Pi)| \le w$ for each $x \in X$ and therefore $k_{\max} \leqslant w$.\qed
\end{proof}

\section{Concluding Remarks}

Cryptographic enforcement schemes (CESs) fall into two broad categories: those that use symmetric cryptographic primitives and those that use asymmetric ones (notably attribute-based encryption~\cite{BeSaWa06,OsSaWa07}).
The focus of this paper is on symmetric schemes, which may be characterized by%
\begin{inparaenum}[(i)]
 \item the total number of secrets required,
 \item the number of secrets required per user,
 \item the total amount of public information required for the derivation of secrets, and
 \item the number of derivation steps required.
\end{inparaenum}

Until recently, symmetric CESs for information flow policies have assumed each user would be given a single secret, from which other secrets and decryption keys would be derived using public information generated by the scheme administrator (see, for example,~\cite{atal:dyna09,cram:key06}).
In this setting, there is a considerable literature on the trade-offs that are possible by reducing the number of steps required for the derivation of secrets, at the cost of increasing the amount of public information (see, for example,~\cite{AtBlFr07,Cr11,ArSaFeMa09}).

One drawback of these types of CESs is that the administrator must generate and publish information to facilitate the derivation of secrets (and decryption keys).
Moreover, the amount of public information required may be substantial, particularly when security labels are defined in terms of (subsets of) attributes.
Chain-based CESs obviate the requirement for public information, the trade-off being that each user may require several secrets. 
The chain-based approach may well be much more practical, particularly if the poset is large and its Hasse diagram contains many edges (as in a powerset, for example).
Moreover, chain-based CESs may be implemented using one-way functions, typically the fastest of cryptographic primitives in practice.

However, it was not known which choice of chain partition was most appropriate for a given information flow policy.
% Our work provides formal and practical methods for constructing the best chain partition, in the sense that it will lead to the smallest number of keys in total and require no user to have more than $w$ keys, where $w$ is the width of the information flow policy.
Our work provides formal and practical methods for constructing a chain partition with the smallest number of keys in total, with the additional property that no user is required to have more than $w$ keys, where $w$ is the width of the information flow policy.

One question remains: If there exist multiple chain partitions that minimize the number of keys in total and per-user, which of these should we choose and can we compute it efficiently?
The one parameter that our work does not address is the number of derivation steps $d$ required by a user in the worst case.
% If we compare the chain partitions in Fig.~\ref{subfig:chainpartition13keys}, for example, we see that a user assigned to the maximum node will need to perform at most $3$ derivations.
% Moreover, for any other chain partition of cardinality $2$ there must be a chain of cardinality $4$ (since the cardinality of the poset is $8$).
% Thus, in this case, the chain partition in Figure~\ref{subfig:chainpartition13keys} is optimal.
% Our future work, then, will attempt to find an efficient algorithm that takes a poset as input and outputs a chain partition into $w$ chains that minimizes $\widehat{K}$ and minimizes $d$.
Our future work, then, will attempt to find a polynomial-time or fixed-parameter algorithm that takes a poset as input and outputs a chain partition into $w$ chains that minimizes $d$.
We also hope to investigate whether the insight provided by Lemma~\ref{lem:number-of-keys-from-bottom-elements}---that $\widehat{K}(\Pi)$ is completely determined by the bottom elements in $\Pi$---can be exploited to design an algorithm whose performance improves on that of the algorithm described in Section~\ref{sec:chain-partition-requiring-widehat-k-keys}.
% The naive algorithm that examines all subsets of cardinality $w$ has worst-case complexity $O(n^w)$; this may well be acceptable in practice, given that $w$ is likely to be small.

\paragraph{Acknowledgements.}
The authors would like to thank Betram Poettering for his valuable feedback and the reviewers for their comments.
\nocite{CiViFoJaPaSa10}

\bibliography{refs}
\bibliographystyle{splncs03}

\end{document}